\pdfoutput=1

\documentclass{cccg25_hidden}
\usepackage{graphicx,amssymb,amsmath}
\usepackage{algorithm}
\usepackage[noend]{algorithmic}
\usepackage[hidelinks]{hyperref}
\usepackage{tikz}
\usepackage[export]{adjustbox}
\usetikzlibrary{decorations.pathreplacing, calc}

\newcommand{\lrParents}[1]{\left( #1 \right)}
\newcommand{\Oh}[1]{O\!\lrParents{#1}}
\newcommand{\OTheta}[1]{\Theta\!\lrParents{#1}}
\newcommand{\dist}[1]{\textnormal{dist}\!\lrParents{#1}}
\newcommand{\R}{\mathbb{R}}
\newcommand{\eps}{\varepsilon}

\newcommand{\T}{\mathcal{T}}
\newcommand{\LL}{\mathcal{L}}
\newcommand{\G}{\mathcal{G}}
\newcommand{\Z}{\mathcal{Z}}

\newcommand{\algoTwoRuntime}{nk \log (nk) \frac{1}{\eps}\left(\log(\Delta_{\T}/\max \{r^*,\rho\}) + k \right)}

\title{Trajectory Minimum Touching Ball}

\author{Jeff M. Phillips\thanks{School of Computing, University of Utah, USA \texttt{jeffmp@cs.utah.edu}. Jeff M Phillips thanks support from NSF 2115677, 2311954, and NIH R37CA276365.} \and Jens Kristian Refsgaard Schou\thanks{Department of Computer Science, Aarhus University, Denmark \texttt{jkrschou@gmail.com}. Jens Kristian R. Schou thanks support from Independent Research Fund Denmark (DFF), grant~9131-00113B}}

\index{Phillips, Jeff M}
\index{Schou, Jens Kristian R}

\begin{document}
\thispagestyle{empty}
\maketitle

\begin{abstract}
    We present algorithms to find the minimum radius sphere that intersects every trajectory in a set of $n$ trajectories composed of at most $k$ line segments each. When $k=1$, we can reduce the problem to the LP-type framework to achieve a linear time complexity.  For $k \geq 4$ we provide a trajectory configuration with unbounded LP-type complexity, but also present an almost $O\left((nk)^2\log n\right)$ algorithm through the farthest line segment Voronoi diagrams.  If we tolerate a relative approximation, we can reduce to time near-linear in $n$.  
\end{abstract}

\section{Introduction}
\label{sec:intro} 

A spatial trajectory is one of the most common non-trivial spatial geometric object (after simplistic points).  It can capture the movement path of people~\cite{sila2016analysis}, wild animals~\cite{kranstauber2011movebank}, vehicles~\cite{matheny2020scalable}, and other objects in a spatial domain; and collecting such data has, in the past decade, has become exponentially easier due to the proliferation of cheap mobile devices with reliable GPS, batteries, and either internal storage or internet connections for cloud storage.  But analyzing such data can be a tangled mess.  

This paper provides algorithms and analysis for one such natural challenge.  Consider a set of identified trajectories $\T$, of which we want to investigate. 
For instance, these could be people who's cell phones were infected by a compromised WiFi router, or animals that got sick from an unknown watering hole, or vehicles that need an electric charging station.  If the spatial event which caused such trajectories to be flagged is unknown, then a goal is to find the most likely region.  We define this (formalized mathematically below) as the smallest (circular) region which all marked trajectories in $\T$ passed through.

In particular, this paper formalizes this problem, and provides 3 algorithmic results. 
\begin{itemize}
    \item Given that the goal region is a convex disk, we ask if we can formulate this as an LP-type problem.  We show that this only works for simple trajectories which are $1$ segment long.  
    \item Next we consider any exact algorithms.  We reduce the furthest point Voronoi diagram, and provide $\Oh{(\alpha(n)+k)n^2k\log n}$ time algorithm in $\R^2$ where there are $n = |\T|$ trajectories, and each is at most $k$ segments, and $\alpha(\cdot)$ is the inverse Ackermann's function.  
    \item  Finally, we consider approximation algorithms in $\R^2$, and show how to find a ball which intersects all trajectories and $(1+\eps)$-approximates the optimal radius up to a minimum added error $\rho$ in $\Oh{\algoTwoRuntime}$ time, where $\Delta_{\T}$ is the diameter of all points in $\T$. 
\end{itemize}

\section{Preliminaries}
\label{sec:prelim}

We encode a trajectory $T$ as an ordered set of $k+1$ waypoints $p_1, p_2, \ldots, p_{k+1} \in \R^d$.  These in turn define an ordered set of $k$ connected line segments $\ell_1, \ell_2, \ldots, \ell_k$ where $\ell_j = \overline{ p_j p_{j+1} }$.
In particular, we can parameterize a point on each segment as $\ell_j(\lambda) = (1-\lambda)p_j + \lambda p_{j+1} \in \R^d$ for $\lambda \in [0,1]$.  
We then represent the full trajectory $T$ as the union of all points $\ell_j(\lambda)$ for $\lambda \in [0,1]$ and $j \in 1,2,\ldots,k$.    

We can handle trajectories of various lengths $k$.  However, for notational convenience, we typically restrict our discussion to a set of trajectories $\T$ where each $T \in \T$ has the same number of segments $k$.  When $k=0$, all trajectories are points.  When $k=1$, all trajectories are line segments.

\paragraph{The minimum touching ball.}
Recall a ball $B_r(c) = \{ x \in \R^d \mid \|x-c\| \leq r\}$ is all points within Euclidean distance or radius $r$ of center point $c \in \R^d$.  

Given a geometric object (a closed set) $Z \subset \R^d$, we say it intersects, or \emph{touches}, a ball $B_r(c)$ if there exists a point $x \in Z \cap B_r(c)$.  Given a set $\Z$ of geometric objects, the \emph{minimum touching ball (MTB)} $B_{r^*}(c^*)$ is the minimum radius ball that touches all $Z \in \Z$.  
Note that there may be multiple balls with the same minimal radius $r^*$ (e.g., when $\Z$ is comprised of two parallel line segments).  Thus, it is technically a \emph{minimal} touching ball, but we usually discuss it as if it is unique.  

We are particularly interested in the MTB problem, where $\Z$ is a set $\T$ of trajectories of length $k$.  We call this the \emph{trajectory minimum touching ball (TMTB)} problem.  Let $r^*(\T)$ is the minimal radius touching ball for $\T$; often we use $r^*$ when the context is clear.  

We also consider and approximate version of this problem.  For $\varepsilon, \rho \geq 0$, the \emph{$(\varepsilon,\rho)$-approximate TMTB} is a ball $B_r(c)$ which touches all $T \in \T$, and $r \leq (1+\eps)\max \{r^*(\T), \rho\}$.  That is, if the minimal radius $r^*(\T) \geq \rho$, then the goal is to find a $(1+\eps)$-relative error approximation and if not to find a radius $\rho$ MTB.

\section{Reduction to LP-Type Problems}
\label{sec:LPt}

An LP-type problem~\cite{lptype}, is a combinatorial optimization problem that can in many ways, especially in low dimensions, be solved with the same algorithms as linear programming.  
It takes as input a set of constraints $S$, and an objective function $f : 2^S \to \R$ satisfying two axioms: \emph{monotonicity} and \emph{locality}.  These are defined with respect to any nested sequence of constraints $B \subset Y \subset S$, and then any particular constraint $x \in S$:  
\begin{itemize}
    \item \underline{Monotonicity:} $f(B) \leq f(Y)$
    \item \underline{Locality:} $f(B \cup \{x\}) > f(B) = f(Y)$ implies \\ \phantom{Locality:} $f(Y \cup \{x\}) > f(Y)$
\end{itemize}

The goal of an LP-type problem is to compute $f(S)$, but this is often opaque given a large set of constraints $S$.  This phrasing is useful when it is efficient for a small set $B$ to compute $f(B)$.  Then if one can identify the smallest set $B \subset S$ such that $f(B) = f(S)$, one can efficiently compute $f(S)$.  
For any $Y \subseteq S$, the minimal set $B \subseteq Y$ with $f(B) = f(Y)$ is called its \emph{basis}.  The maximal possible basis size is called the \emph{combinatorial dimension} of the LP-type problem.  
For example, in linear programming in $\R^d$ then $f(Y) = \max_z \langle u,z\rangle$ and we need to satisfy all linear constraints $x \in Y$ encoded as $\langle \alpha_x,z\rangle \geq \beta_x$.
Here the optimum is defined by at most $d$ constraints, so the combinatorial dimension is $d$.  

Assuming the combinatorial dimension is an absolute constant, several algorithms~\cite{seidel,Clarkson,MatousekSharirWelzl} can compute a minimal basis set $B$ in expected time linear in the number of constraints $|S|$, and thus, compute $f(S)$ in expected linear time in $|S|$.  

Famously, the minimal enclosing ball (MEB) problem~\cite{MatousekSharirWelzl} is LP-type. The constraints are a set of points $x \in \R^d$ which must be contained in a ball $B_r(c)$, and $f(S)$ is the minimal radius for which there exists such a ball that satisfies the constraints $S$. The combinatorial dimension is $d+1$, so in constant dimension $d$, this gives an expected linear time solution in the number of points.  
Specifically, for MEB, linear program $\R^d$ for constant $d$, or any LP-type problem with combinatorial dimension $d$, then a randomized algorithm that combines methods of \cite{Clarkson,Kalai,MatousekSharirWelzl} computes such a solution in $\Oh{d^2n+e^{\Oh{\sqrt{d\log d}}}}$ arithmetic operations, in expectation. 

A simple consequence of this is that for $k=0$, the TMTB problem has trajectories as points, and is the MEB problem.  This gives an expected $\Oh{n}$ time solution in $\R^d$ for constant $d$.  
A natural question, that we answer in the negative in this paper, is if this can extend to the general length $k$ trajectory problem.  

Most optimization problems can be phrased within the LP-type framework. However, in some cases, the minimal basis includes all constraints, resulting in an unbounded dimension. In such cases standard LP-type solver time bounds are exponential in the input size. Thus to invoke LP-type solvers, it is paramount to bound the combinatorial dimension as constant.  

\subsection{Segment TMTB reduction to LP-type}
\label{sec:segmentLP}

As a natural first step we consider trajectories with $k=1$, so each $T \in \T$ is a single line segment. In this setting, the TMTB and the MEB no longer coincide.  

To analyze line segment MTB, we leverage Amenta's connection between LP-type problems (there called generalized LP), and Helly-type problems~\cite{Amenta1994}.  She includes a result that for a family $\mathcal{K}$ of convex objects, and a fixed convex object $C$, finding the smallest homothet of $C$ intersecting every member $K \in \mathcal{K}$ is an LP-type problem. And the basis size is constant if each element $K \in \mathcal{K}$ has constant description complexity.  If these objects are  in $\R^d$, then the combinatorial dimension is $d+1$.   In this framework, we let the homothets of $C$ be the family of all balls, and size parameterized by radius. When each $T \in \mathcal{T}$ is a segment, it is convex so $\mathcal{K} = \mathcal{T}$ and we can conclude the following:

\begin{lemma}
    Line segment MTB in $\R^d$ is an LP-type problem with combinatorial dimension $d+1$.
\end{lemma}

The next natural question is if this argument extends to trajectories with more than one line segment.  
The reduction via Helly's theorem~\cite{Amenta1994} crucially relies on convexity.  But the distance function to trajectories of size $k \geq 2$ is not convex.  
On the other hand, Amenta presents a case of sets $\mathcal{B}_2$ where each $K  \in \mathcal{B}_2$ is two sufficient separated balls (hence $K$ is not convex), yet finding the smallest ball intersecting each $K$ is still LP-type with combinatorial dimension $2d+1$.  So convexity is not \emph{required} for this formulation to work.

\subsection{General TMTB \underline{does not} have Bounded Combinatorial Dimension}
\label{sec:notLP}

When trajectories are allowed to have at least $k=4$ segments, we can construct a configuration in which every trajectory is essential to defining the optimal solution of the TMTB. 
This implies that the LP-type dimension is unbounded, and the LP-type framework does not lead to efficient solutions.  

\begin{lemma}\label{lemma:traj_lp_monster}
    The combinatorial dimension of TMTB with $k=4$, as an LP-type problem, is unbounded.
\end{lemma}
\begin{proof}Let $n > 4$. We construct a set of $ n $ trajectories $ \mathcal{T} $ in $ \mathbb{R}^2 $ such that for every trajectory $ x \in \mathcal{T} $, removing $ x $ strictly reduces the minimum enclosing radius, i.e.,
$$
f(\mathcal{T} \setminus \{x\}) < f(\mathcal{T}),
$$
where $ f(\cdot) $ denotes the TMTB radius function. This guarantees that all $ n $ trajectories must be part of any LP-type basis, implying a combinatorial dimension of at least $ n $.

The construction proceeds as follows; see Figure~\ref{fig:traj_lp_monster} for an illustration.  
First, define a fixed line segment $ T_0 $ that serves as the base. It is slanted and extends from the point $ (0, 0) $ to $ (n + 2.5, -0.5) $. This segment helps anchor the MTB from below.
Next, define a special trajectory $ T_1 $, which starts high and descends toward the base. It begins at $ (0, 4) $, descends through the point $ (3.5, 1) $, and ends at $ (n + 2.5, 1) $. This trajectory stabilizes the top and the right side of the MTB.

The remaining $ n - 2 $ trajectories form a sequence of nested arches. For each $ 2 \leq i \leq n - 1 $, define the trajectory $ T_i $ as:
$$
(0, 1) \rightarrow (i, 1) \rightarrow \left(i + \tfrac{2.5}{2}, 4\right) \rightarrow (i + 2.5, 1) \rightarrow (n + 2.5, 1).
$$
Each such trajectory starts at height 1, rises to height 4, and returns to height 1, forming a peaked shape. These paths create vertical obstructions that prevents a small MTB; a ball with the $x$-coordinate of its center at $i + 2.5$ will be sufficiently far from $T_i$ that it will be beneficial to move to a larger $x$ after all trajectors return from their peaks, as shown in Figure \ref{fig:traj_lp_monster} and Appendix~\ref{apx:lemma2figs}.  

In this setup, every trajectory is essential to defining the optimal MTB. If any trajectory $T_i$ is removed, the touching ball can be given an $x$-value of its center at $i+\frac{2.5}{2}$ and touch both $T_{i-1}$ and $T_{i+1}$ without expanding its radius too much.   In particular, if $T_0$ is removed, all remaining trajectories intersect, and the MTB radius drops to zero. If $T_1$ is removed, the center moves to $(0, 0.5)$ with radius $0.5$. For $i > 1$ each trajectory $T_i$ is carefully placed to block a tighter ball from forming around the previous trajectories. Removing $T_i$ allows the TMTB to be defined by earlier trajectories $T_j$ for $j < i$. The slant of $T_1$ ensures that the MTB in the full configuration has to increase to a larger $x$-value around $n+2.5$. Thus, no trajectory is redundant, and all trajectories must be included in any LP-type basis, making the combinatorial dimension unbounded, as no small subset can represent the entire set.
\end{proof}

\begin{figure}
    \centering
    \includegraphics[trim={2.4cm 2.9cm 1.75cm 2.9cm},clip,scale = 0.7]{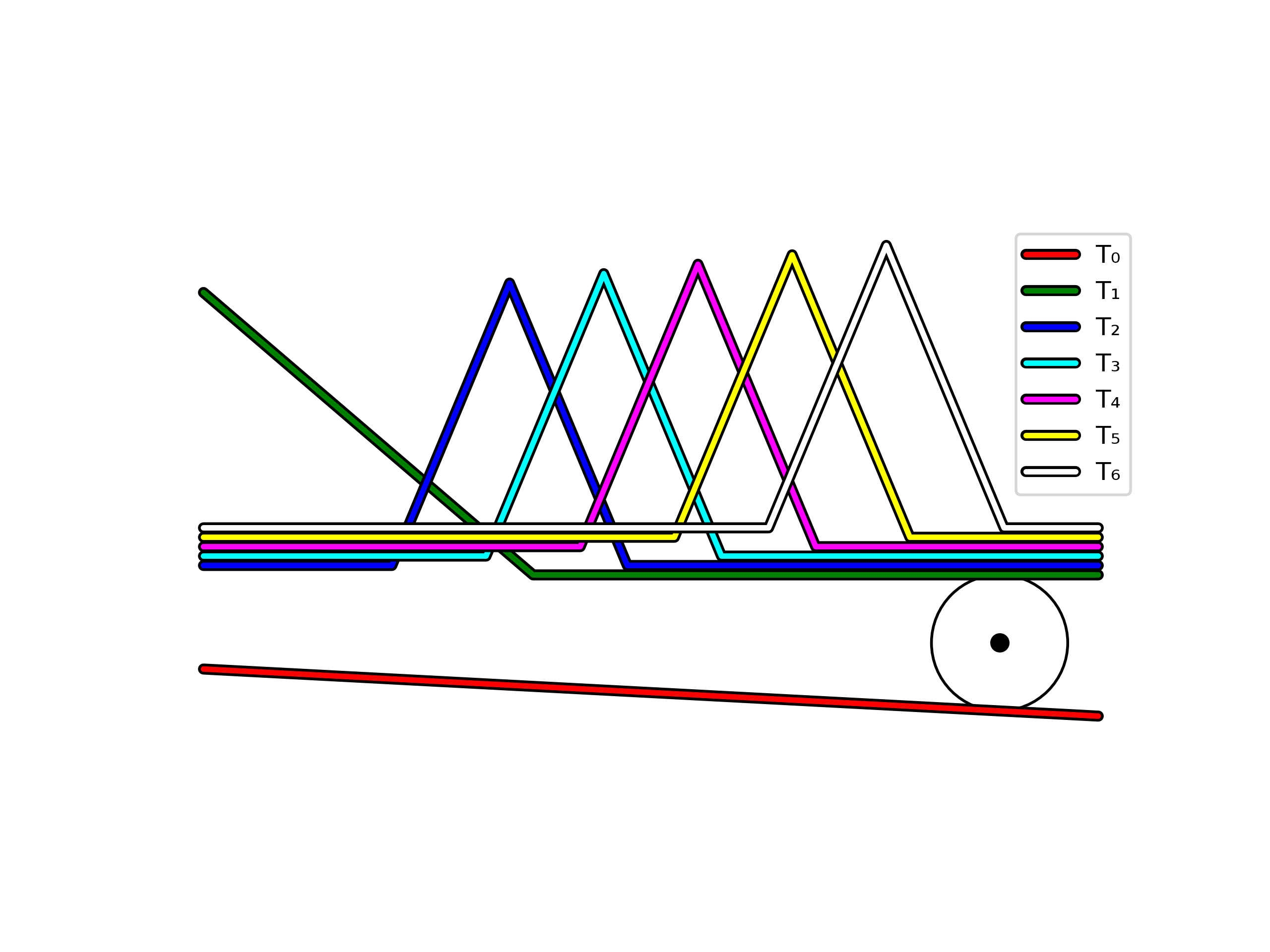}
    \caption{The trajectory configuration of Lemma~\ref{lemma:traj_lp_monster} that has unbounded LP-type dimension, each trajectory has a different color and the TMTB is visualized in black. The trajectories are supposed to be on top of each other on the horizontal segment, but for clarity of presentation, they are raised.} 
    \label{fig:traj_lp_monster}
\end{figure}

\section{Furthest Trajectory Voronoi diagrams}
\label{sec:TVD}
In this section, we use the Farthest-Color Line Segment Voronoi Diagrams (FCLVD) of Bae~\cite{Bae_farthest_voronoi} to define Farthest Trajectory Voronoi Diagrams (FTVD) and show that the optimal TMTB center must be on one of its edges or vertices.  

Bae~\cite{Bae_farthest_voronoi} considered a set of $N$ total line segments $\LL$ in $\R^2$ in $K$ different colors, with the general position assumption that no three segments cross in at a single point. Let the $j$-colored line segments be denoted as $\LL_j$. 
Let 
\[
\dist{x,\LL_j} = \min_{z \in \ell; \ell \in \LL_j} \|x-z\|
\]
be the closest point from query $x$ to some point $z \in \ell$ for any line segment $\ell \in \LL_j$ of color $j$.  
Then the \emph{furthest color line segment Voronoi Diagram (FCLVD)} is a decomposition of $\R^2$ into regions for each color $j$ as $D_j = \{x \in \R^2 \mid \dist{x,\LL_j} > \dist{x, \LL_{j'}}
\text{ for } j'\neq j\}$, so $x$ is further from any point in $\LL_j$ than to any other set of colored segments $\LL_{j'}$.  
Bae~\cite{Bae_farthest_voronoi} showed that if there are a total of $h = \Oh{N^2}$ segment crossings, then the FCLVD can be constructed in time $\Oh{(NK+h)(\alpha(K)\log K + \log N)}$, where $\alpha(k)$ is the inverse Ackerman function.  
The combinatorial complexity of a FCLVD is the sum of the number of cells, edges, and vertices of the diagram.  The \emph{cells} are the maximal connected components among the $D_j$ regions.  The \emph{vertices} are the points $x \in \R^2$ where there are $3$ equally close line segments $\ell_1, \ell_2, \ell_3$; each which could serve as the furthest colored line segments.  It could generally be the case that two of the three segments have the same color.
And the \emph{edges} are the components of the boundaries between two cells $D_j$ and $D_{j'}$; the locus of points $x \in \R^2$ so that $\dist{x, \ell} = \dist{x, \ell'}$ which are the closest line segments from the two furthest colored sets $\ell \in \LL_j$ and $\ell' \in \LL_{j'}$.  The boundaries of the edges are vertices; although not all edges may have (both) boundaries.  
Bae showed the worst-case combinatorial complexity is $\OTheta{NK+h}$.  
In any construction, each vertex, edge, and face can be assigned its \emph{generators}, the $3$, $2$, or $1$ segments (respectively; and its color) which determine its geometry.  

Now we define the \emph{furthest trajectory Voronoi diagram (FTVD)} as an extension of the FCLVD.  Here we consider $n$ trajectories $\T$, each composed of at most $k$ line segment in $\R^2$.  The FTVD is again an decomposition of $\R^2$ so for each $T_j \in \T$, it defines regions $D_j = \{x \in \R^2 \mid \dist{x,T_j} < \dist{x,T_{j'}} \text{ for } j \neq j' \}$.  There are $K = n$ colors (one for each trajectory $T_j$), and at most $N = kn$ total line segments.  Hence the total combinatorial complexity is $\Oh{k^2n^2}$ since there are at most ${nk \choose 2} = \Oh{n^2k^2}$ intersections.  And using Bae's algorithm the FTVD can be constructed in $\Oh{(\alpha(n) +k)n^2 k \log n}$ time.  
Figure~\ref{fig:traj_FTVD} shows a color-coded FTVD with the TMTB in black.  

\begin{figure}
    \centering
    \includegraphics[trim={2.4cm 3.1cm 2cm 2.9cm},clip,scale = 0.7]{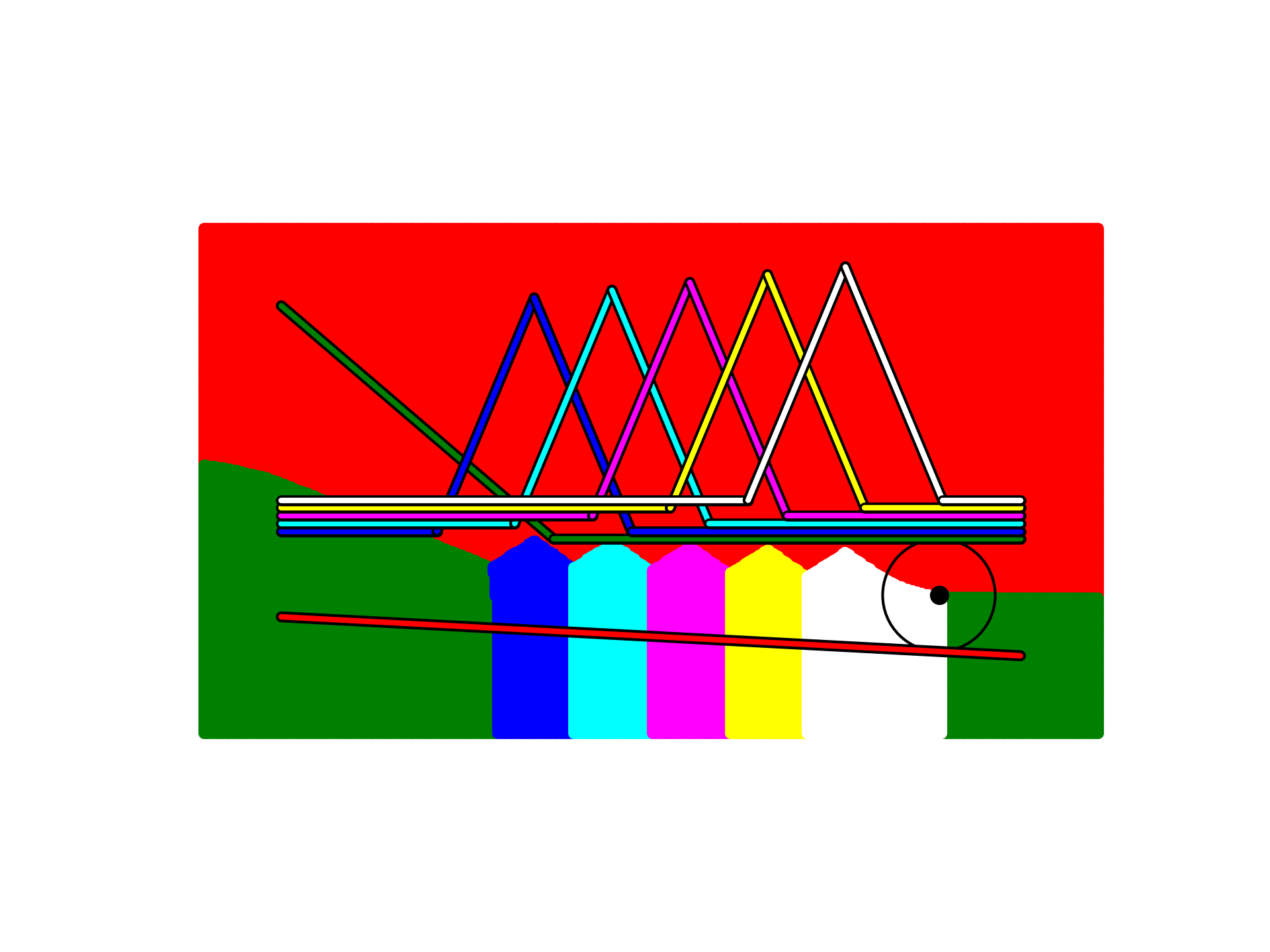}
    \caption{Farthest Trajectory Voronoi Diagram of the lower-bound construction presented in Lemma~\ref{lemma:traj_lp_monster}, note that the TMTB is on a curved edge of the diagram - in line with Theorem~\ref{thm:FTVD=MTB}}
    \label{fig:traj_FTVD}
\end{figure}

\begin{theorem}\label{thm:FTVD=MTB}
    Consider $n$ trajectories with at most $k$ segments in $\R^2$ with no 3 segments crossing at a single point, then the TMTB can be found in $\Oh{(\alpha(n) +k)n^2 k \log n}$ time.
\end{theorem}
\begin{proof}
We proceed by leveraging the structure of the FCLVD applied to the FTVD.  We first argue that the center $c$ of the optimal TMTB for $\T$ cannot be in (the interior of) a face $D_j$ of the FTVD.  If it where, then $T_j$ is strictly further from $c$ than any other trajectory $T_{j_1}$ and the radius of the ball is $r = \dist{c,T_j} > 0$.  Hence, we can move $c \to c'$ infinitesimally towards the closest point on $T_j$.  The new center $c'$ is still in the face, hence $\dist{c',T_j} < \dist{c,T_j}$ and $\dist{c', T_j}  = r' < r$ but $c$ still determines the radius of the ball while the radius of the ball is smaller than before, a contradiction.

Now we can reduce the TMTB problem to considering centers $c$ which live on a vertex or edge of the FTVD.  
Moreover, we know that if $c$ is the center of the MTB, and it lies on an edge or vertex, then the radius of the ball is the distance to the generators (segments $\ell_1$, $\ell_2$, and maybe $\ell_3$) of the vertex of edge.  This holds since we know the distance $\dist{c,\ell_j} = r$ for $j \in \{1,2,3\}$ is the distance to the furthest trajectory, which means that all other trajectories $T_{j'}$ are within distance $r$; and if we decrease $r$ then those trajectories associated with the generators are no longer touched by the ball. 

Hence, we can start by considering all, at most $\Oh{k^2n^2}$, vertices of the FTVD as potential centers of the TMTB.  We only need to check their $3$ generating line segments for the required radius, which takes $\Oh{1}$ time each.  The smallest radius is a valid touching ball, and an upper bound on the minimum touching ball.  

It remains to consider centers on the edges of the FTVD.  Again these can each be checked in $\Oh{1}$ time using its two generators, but requires more care as we provide next.  An edge of the FTVD (and by Bae's analysis~\cite{Bae_farthest_voronoi} the FCLVD) is formed in three structural ways: the interior of both segments, vertices of both segments, or an interior of one segment and a vertex of the other.  
In each case, if the center is on the FTVD edge, it will be at the location $c$ where the distance to the generators is minimal.  
When the edge is formed by two endpoints, the edge is their bisector, and the minimum is obtained at the midpoint between them.  
When the edge is formed from the interior of two segments, the distance is decreasing in one direction along the edge as the segments are getting closer; if the minimum is on this edge, it will be at its endpoint, which may be a vertex.  This holds unless the segments are parallel, in which case any point on the edge can serve to define the radius, as half the distance between parallel segments.  
When the edge is formed by one endpoint and one segment interior, the edge is curved, but the distance is minimized on the edge at the midpoint between the endpoint and its projection on the segment interior.  
In each of these cases, the distance to the generators decreases monotonically as we move along the edge to the minimum. 
Note that the minimum touching ball of two generators may not have its center on the edge associated with those generators (e.g., it may be a vertex, or another pair); this implies the other generators have a smaller radius, and are also valid touching ball, dominating this one.  Hence we do not need to explicitly exclude such cases.

In conclusion, the center of the TMTB must occur on an edge or vertex of the FTVD.  If the center occurs on one of those objects, we can determine its location and the associated radius in $\Oh{1}$ time, given the generators associated with the FTVD object.  There are at most $\Oh{k^2n^2}$ FTVD objects (vertices and edges) to check, and the FTVD can be constructed in $\Oh{(\alpha(n)+k)kn^2\log n}$ which bounds the runtime. 
\end{proof}

\section{Approximate Trajectory Minimum Touching Ball}
\label{sec:approx}

Given a point $ c $, we can construct a TMTB candidate for $ \mathcal{T} $ by computing its radius as 
$
r = \max_{T \in \mathcal{T}} \operatorname{dist}(c, T),
$
where, recall, distance is defined as 
$
\dist{c, T} = \min_{\ell \in T} \min_{x \in \ell} \|c - x\|.
$
This can be computed in $O(k)$ time using straightforward iteration over segments, or in $ \Oh{k \log k} $ preprocessing and $ \Oh{\log k} $ query time using a Closest Segment Voronoi Diagram (CSVD), due to~\cite{CLVD}. In this section, we present results that depend on the geometry of the trajectories. Let $r^*$ denote the optimal radius of a TMTB of the trajectories $\T$, and let 
$
\Delta_{\T} = \max_{\substack{x \in T \in \T \\ x' \in T' \in \T}} \|x-x'\|
$
be the diameter of $\T$.  

The simplest way to approximate a TMTB is to define a uniform grid over the bounding box of the trajectories and evaluate the maximum distance from each grid point to all trajectories. A grid of width $\eps$ contains $\Oh{(\Delta_{\T} / \eps)^2}$ points. Evaluating each point takes $O(n \log k)$ time, where $n$ is the number of trajectories and each has $k$ segments. This yields a total runtime of $\Oh{(\Delta_{\T} / \eps)^2 \cdot n \log k }$ for a TMTB approximation with additive error $\eps$.

We present an algorithm for a $(\eps,\rho)$-approximate TMTB in two stages: we first obtain a constant approximation via a reduction to a constant-factor approximation, then refine this to get a $(\eps,\rho)$-approximate TMTB. Our key idea is to assume that the center lies somewhere along a trajectory $T$, and solve the decision problem if the radius is less than a threshold $\tau$.  

\subsection{Constant Factor Approximate TMTB}
Given a radius estimate $\tau$, we define, for each segment $\ell \subseteq T$ and each trajectory $T \in \mathcal{T}$, the set
\[
I_T^{\tau} = \{x \in \ell \mid \operatorname{dist}(x, T) \le \tau \},
\]
i.e., the portion of $ \ell $ that is within distance $ \tau $ from $ T $. Taking the intersection over all trajectories gives the segment-wise feasible region:
\[
F_T^{\tau} := \bigcap_{T \in \mathcal{T}} I_T^{\tau}.
\]
If this intersection is non-empty for some segment $\ell \subseteq T $, then any point in it defines a center with radius at most $ \tau $ that touches every trajectory in $\T$. 

To refine the estimate, we perform a geometrically decreasing  search: we shrink $\tau$ by a factor $\gamma$, and recompute $F_t^{\tau}$ until the intersection becomes empty, or we reach a precision threshold $\rho$. If the intersection vanishes, we move to the next segment $\ell$ of $T$.  Initializing $\tau = 2 \hat \Delta$ as a $2$-approximation of $\Delta_\T$ ensures we start with a valid upper bound, and is sketched in Algorithm \ref{algo:binary_search_radius}: \textsc{EstimateRad}.

\begin{algorithm}
\caption{\textsc{EstimateRad}($\T, T, \gamma, \rho, \tau$).}
\label{algo:binary_search_radius}
\begin{algorithmic}[1]
    \FOR{each segment $\ell$ of $T$}
        \WHILE {($\tau > \rho$)}{
            \FOR{each trajectory $T' \in \T \setminus \{T\} $}
                \STATE Compute intervals $I_{T'}^\tau\subseteq\ell $ 
            \ENDFOR
          \STATE Compute the intersection $F_\ell^\tau = \bigcap\limits_{T'\in\T\setminus \{T\}} I_{T'}^\tau$ 
          \STATE \textbf{if} $F_\ell^\tau = \emptyset$ \textbf{break} \emph{[go to next $\ell \in T$]}
          \STATE \textbf{else} $\tau \leftarrow \tau/\gamma$
        }\ENDWHILE
    \ENDFOR
    \RETURN{$\gamma \cdot \tau$}
\end{algorithmic}
\end{algorithm}

\begin{lemma}\label{lemma:4-approx}
    \textsc{EstimateRad}($\T,T,\gamma,\rho, 2\hat \Delta)$ (Alg~\ref{algo:binary_search_radius}) computes a $(\gamma(1+\beta)-1,\rho)$-approximate TMTB with center on $T$, where $\beta$ is such that $\dist{c^*(\T),T} \leq \beta r^*(\T)$, 
    in time 
    \[
O(nk \log(nk) (\log_\gamma(\Delta_{\T} / \max\{r^*,\rho\}) + k)).
    \]
\end{lemma}
\begin{proof}
The runtime analysis is straightforward. 
We can compute the $2$-approximation $\hat \Delta$ in $O(nk)$ time by choosing any point $x \in T \in \T$ and for each segment $\ell \in T' \in \T$ 
finding the furthest point from $x$; we set $\hat \Delta$ as the max distance among these options.  

Notice that $I_{T}^\tau$ (line 4) consists of at most $k$ intervals on the segment $\ell$ that each can be constructed in constant time, resulting in $\Oh{nk}$ total intervals.  Their intersection $F_\ell^\tau$ (line 5) can be computed by sorting along $\ell$ in $\Oh{nk\log(nk)}$ time.   
As the algorithm iterates over segments $\ell \in T$, it maintains the smallest threshold $\tau$ for which it found an intersection.  If this shrinks $\tau$ (line 7) $m$ times, then the feasible region $F_{\ell}^\tau$ (line 5) is computed $k + m$ times, since it hits the break statement (line 6) at most $k$ times.
Starting at $\tau = 2 \hat \Delta \geq \Delta_{\T}$, the shrinking terminates when $\tau$ goes below $\rho$ or goes below the optimal radius $r^*_T$ for a touching ball of $\T$ with center on $T$; note $r^*_T \geq r^* = r^*(\T)$. 
Combining these analysis together achieves the claimed runtime.  

To argue for correctness, first note that $r^*  \leq\Delta_{\T} \leq 2\hat \Delta$, since a point at one of the diameter elements is within a radius $\Delta_{\T}$ of all points in $\T$ and hence is a touching ball.  
Then we restrict to centers along a given trajectory $T$, and return $\tau$ such that the optimal radius $r^*_T$ for a touching ball of $\T$ with center on $T$ satisfies $\max \{\tau / \gamma,\rho\} \leq r^*_T \leq \tau$, making $\tau$ a relative $\gamma$-approximation up to distance $\rho$ for $r^*_T$. 

We also know that $r^*_T \geq r^*$ since it is valid, but has more restrictions.  By assumption the optimal center $c^* = c^*(\T)$
satisfies $\dist{c^*,T} \leq \beta r^*$.  Let $c_T$ be the closest point to $c^*$ on $T$, and we know that the radius $r^*_T$ of the TMTB at $c_T$ satisfies $r_T \geq r^*_T$.  Since moving $c^*$ to $r_T$ by $\beta r^*$ can increase the radius by at most $\beta r^*$, then $r^*_T \leq r_T \leq r^*(1+\beta)$.  Putting it all together
\[
\tau/(\gamma(1+\beta)) \leq r^*_T/(1+\beta) \leq r^* \leq r^*_T \leq \tau
\]
Thus $\tau$ is a relative $\gamma(1+\beta)$-approximation of $r^*$, assuming the radius is at least $\rho$; so a $(\gamma(1+\beta)-1,\rho)$-approximation, as claimed.  
\end{proof}

\begin{cor}
    \label{cor:4-apx}
    \textsc{EstimateRad}($\T,T,2,\rho, 2\hat \Delta)$ (Alg~\ref{algo:binary_search_radius}) where $T \in \T$, computes a $(3,\rho)$-approximate TMTB with center on $T$, and 
    in time 
    $
O(nk \log(nk) (\log(\Delta_{\T} / \max\{r^*,\rho\}) + k)).  
    $
\end{cor}
\begin{proof}
To get the relative error analysis, we use $\dist{c^*,T} \leq r^*$ for $T \in \T$, and thus $\beta = 1$.  With $\gamma=2$, then $\gamma(1+\beta)-1 = 3$ as desired.  
\end{proof}
\subsection{$(1+\varepsilon)$-Approximate TMTB}
\label{sec:1+eps-apx}
Given a constant factor estimate $\tau$ of the radius of the TMTB, we construct the $\tau$-\emph{sausage} of trajectory $T$, defined as the set
\[
S_{T,\tau} = \{x \in \mathbb{R}^2 \mid \dist{x, T} \le \tau \}.
\]
The region is the Minkowski sum of $T$ with a ball of radius $\tau$; it is all points within $\tau$ from $T$.  

\begin{lemma} 
\label{lem:sausage}
Let $B_{r^*}(c^*)$ be the TMTB for $\T$ and $T \in \T$.  Then $c^* \in S_{T,r^*}$.  
\end{lemma}
\begin{proof}
    By definition, the TMTB $B_{r^*}(c^*)$ of $\T$ must satisfy $\dist{c^*,T}\leq r^*$.  Since $S_{T,r^*}$ contains all points $x$ satisfying $\dist{x,T}\leq r^*$, that includes $c^*$.  
\end{proof}

Now we use that we have a value $\tau$ (via Lemma \ref{lemma:4-approx}) that is the radius of a $(3,\rho)$-approximation of the TMTB, i.e. $r^* < \max\{ 4 \tau, \rho\}$ and $\tau/4\leq r^*$.  Note that $\tau > r^*$ and $S_{T,r^*} \subset S_{T,\tau}$, so by Lemma \ref{lem:sausage} we have that $c^* \in S_{T,\tau}$.  

To achieve a $(\eps,\rho)$-approximation of the TMTB, we can invoke Algorithm \ref{algo:binary_search_radius} with $\gamma = 1+\eps/3$ and on some $T$ so that $\beta = \eps/3$.  Then, if $\eps < 1/2$, then $\gamma(1+\beta)-1 = (1+\eps/3)(1+\eps/3)-1 \leq \eps$.  

We could for instance run this on a point set $P$ (each $p \in P$ a length-0 trajectory), then we would require sufficiently dense points $P \subset S_{T,\tau}$, so that any point $x \in S_{T, \tau}$ there exists a point $p \in P$ with $\|x-p\| \leq \max\{\rho, \eps \tau/12\} \leq \max\{\rho, \eps r^*/3\}$.  

However, this would require many points in $P$ as we do not control for the length of the trajectory segments relative to $\tau$ and $r^*$.  

Instead we introduce the \emph{$(\eps, \tau)$-ghost trajectories} $\G_{T,\tau,\eps}$ as a collection of trajectories that cover the $\tau$-sausage of $T$, evenly spaced at distance $\tau\eps/12$, and each running parallel to $T$.
We construct the $i$th ghost trajectory as follows. For every segment $\ell$ of $T$, create a parallel segment offset by a distance of $i\tau\eps/12$. At each vertex where two segments meet, we extend/reduce the offset segments until they intersect, and the last ones we extend by $\tau$ so they fill the caps of the sausage. The resulting ghost segment intersection points lie on the angular bisector of the separating vertex.

We include one such ghost trajectory for each admissible offset $|i\tau\eps/12| \leq \tau$, so with integers $i \in [-12/\eps, 12/\eps]$, on both sides of $T$;  see Figure~\ref{fig:sausage_grid}.
This results in a total of $\Theta(1/\eps)$ ghost trajectories, each with (at most) the same number of segments as $T$; some segments will disappear around consecutive same-direction bends. Together they form an $(\tau\eps/12)-$net of the sausage in the sense that each point $x \in S_{T,\tau}$ satisfies $\min_{G \in \G_{T,\tau,\eps}} \dist{x,G} \leq \tau \eps / 12 \leq r^* \eps/3$.  

\begin{figure}
    \centering
    \includegraphics[trim={0.2cm 2.3cm 0.2cm 1.7cm},clip,width=\linewidth]{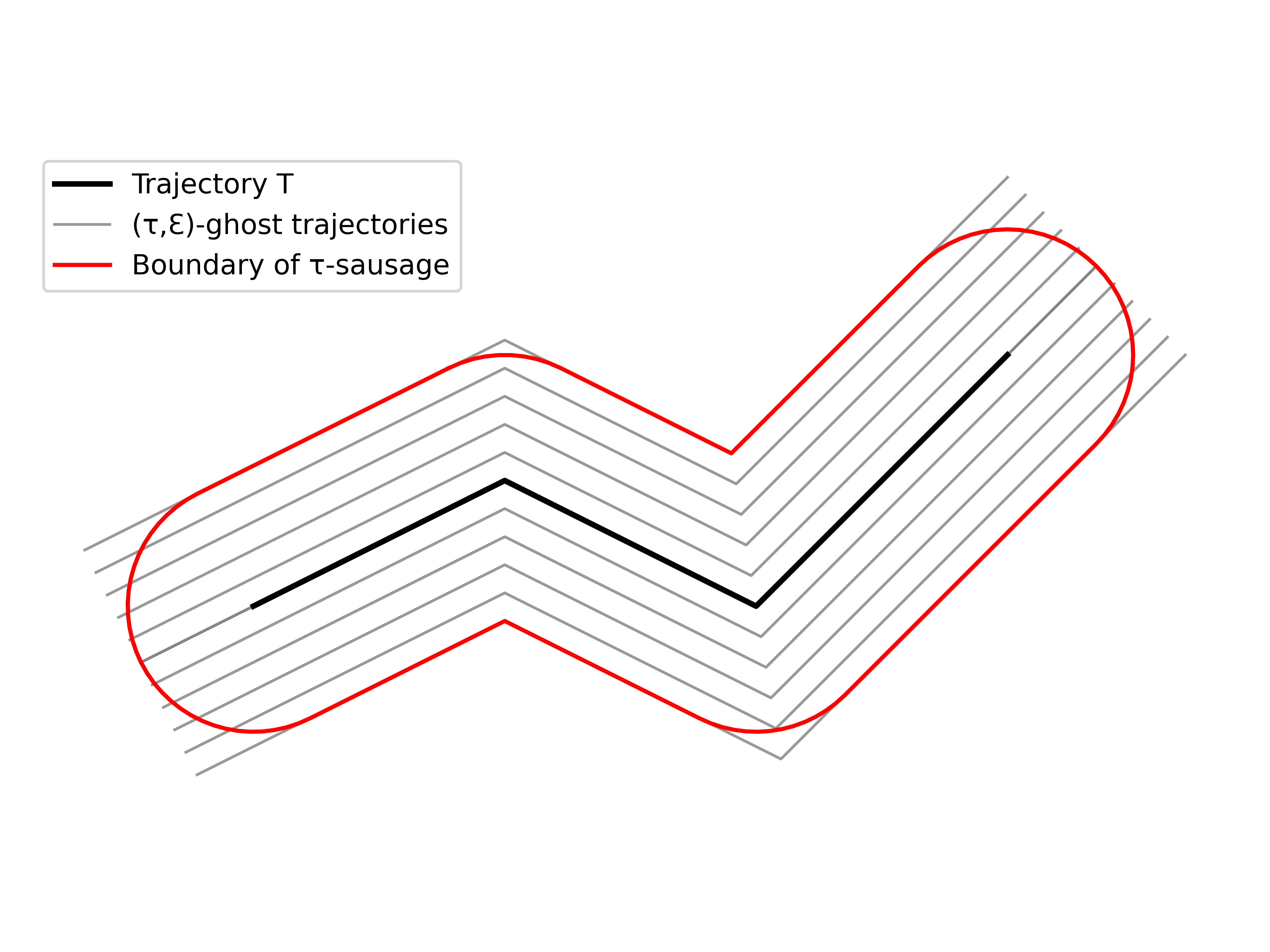}
    \caption{A trajectory sausage $S_{T,\tau}$ covered by ghost trajectories $\G_{T,\tau,\eps}$.}
    \label{fig:sausage_grid}
\end{figure}

The process of checking each ghost trajectory with \textsc{EstimateRad} is outlined in Algorithm \ref{algo:ghost-search}.  
Note that the trajectory $T$ that \textsc{EstimateRad} takes as input is simply used as a set of segments, and so we can use all segments in $\G_{T,\tau,\eps} = \cup_i G_i$ in that role.

\begin{algorithm}
\caption{\textsc{EstimateTMBT}($\T, \eps, \rho$).}
\label{algo:ghost-search}
\begin{algorithmic}[1]
    \STATE Choose any $T \in \T$ and run \textsc{Estimate-Rad}($\T,T,2,\rho,2\hat \Delta)$ to get an estimate $\tau$ of $r^*$.  
    \STATE Compute $O(1/\eps)$ ghost trajectories $\G_{T,\tau,\eps}$ to $(\eps \tau/12)$-cover the sausage $S_{T,\tau}$.  
    \RETURN \textsc{EstimateRad}($\T,  \G_{T,\tau,\eps},\eps/3,\rho,\tau$) 
\end{algorithmic}
\end{algorithm}

\begin{theorem}\label{thm:algo2}
    \textsc{EstimateTMBT}($\T, \eps, \rho)$ (Alg \ref{algo:ghost-search}) for $\eps \leq 1/2$ computes a $(\eps,\rho)$-approximate TMTB of $\T$ in time
\[
\Oh{\algoTwoRuntime}.
\]         
\end{theorem}
\begin{proof}
    By Corollary \ref{cor:4-apx} step 1 is within the time bound, and step 2 takes $O(k/\eps)$, which is also within the bound.  
    Step 3 dominates the cost as it needs to run Algorithm \ref{algo:binary_search_radius} on a set of $\Oh{k/\eps}$ segments, with $\gamma = 1+\eps/3$.  The $\gamma =1+\eps/3$ affects the runtime as $O(\log_{1+\eps/3}(X)) = O(\frac{1}{\eps} \log(X))$.  We maintain the upper bound $\tau$ as we iterate through the $O(k/\eps)$ ghost trajectories, and so the feasible sets are built at most $O(\frac{1}{\eps}\log(\Delta_{\T}/\max\{r^*,\rho\}) + k/\eps)$ times.  
    The claimed runtime follows.

    The accuracy follows by invoking Lemma \ref{lemma:4-approx} using $\gamma=1+\eps/3$ and $\beta = \eps/3$, so $\gamma(1+\beta)-1 \leq 3$ for $\eps < 1/2$.  By the construction of the ghost trajectories, each point $x \in S_{T,\tau}$ satisfies $\dist{x,G} \leq \eps \tau/12 \leq \eps r^*/3$ for some $G \in \G_{T,\eps,\tau}$, especially for the true TMTB center $c^*\in S_{T,\tau}$ due to Lemma~\ref{lem:sausage}.  
    \end{proof}

\section{High-dimensional TMTB}
Our LP-type algorithms for $k=0$ and $k=1$ trajectories, that is, points and line segments, extend to $\R^d$ without modification. For general trajectories, our approximation algorithm extends to $\R^d$ (for constant dimension $d$), by adding $O(1/\eps^{d-1})$ ghost trajectories to cover the sausage $S_{T,\tau}$.  The feasible set is built between pairs of low-dimensional objects and the complexity does not change.  The runtime becomes $O(n\frac{k\log(nk)}{\eps}(\log(\Delta_{\T}/\max \{r^*,\rho\}) + k/\eps^{d-2}))$.


\small

\newpage
\appendix
\section{Figures for Lemma~\ref{lemma:traj_lp_monster} and Theorem~\ref{thm:FTVD=MTB}}\label{apx:lemma2figs}
In this section, we present visual aids to the proof of Lemma~\ref{lemma:traj_lp_monster} and Theorem~\ref{thm:FTVD=MTB} by showing how the construction needs every Trajectory in an LP-type basis and that the TMTB lies on an FTVD vertex.

First, by removing $T_0$, the remaining trajectories intersect, making the problem trivial. Second, by removing $T_1$, the TMTB moves to the left, where the slant of $T_0$ ensures its uniqueness. See Figure~\ref{fig:minus1}.

Now for $n-1 > i > 1$ removing $T_i$ moves the TMTB to where $T_{i-1}$ descends to 'ground level' and $T_{i+1}$ ascends - this is ensured by carefully choosing the width of the spikes. See Figures~\ref{fig:minus2}, \ref{fig:minus3}, \ref{fig:minus4}, and \ref{fig:minus5}.

When removing $T_{n-1}$ the TMTB is localized where $T_{n-2}$ descends. See Figure~\ref{fig:minus6}.

\begin{figure*}
    \centering
    \includegraphics[trim={2.4cm 2.9cm 1.85cm 2.9cm},clip,scale = 0.7]{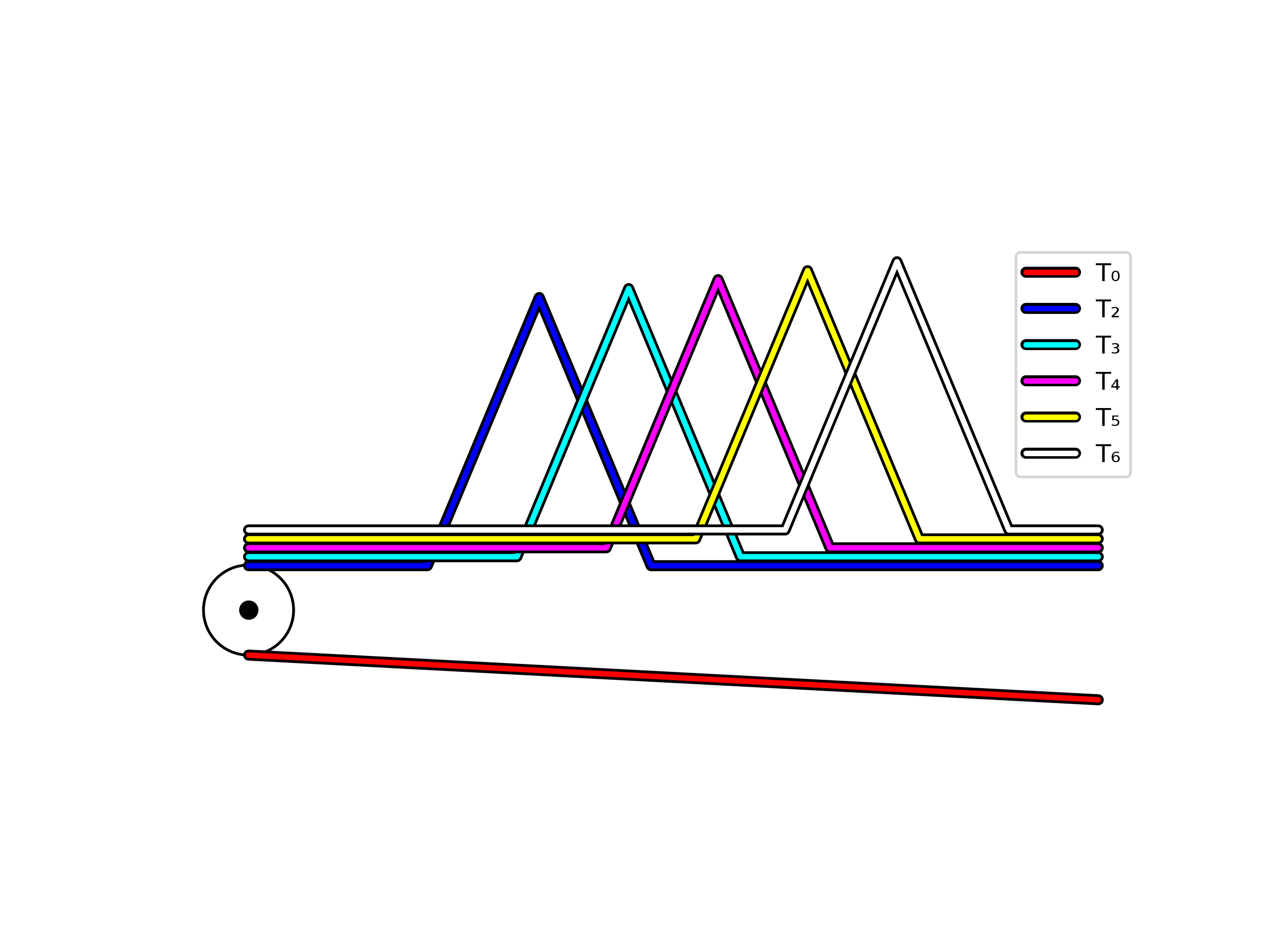}
    \includegraphics[trim={2.4cm 2.9cm 2cm 2.9cm},clip,scale = 0.7]{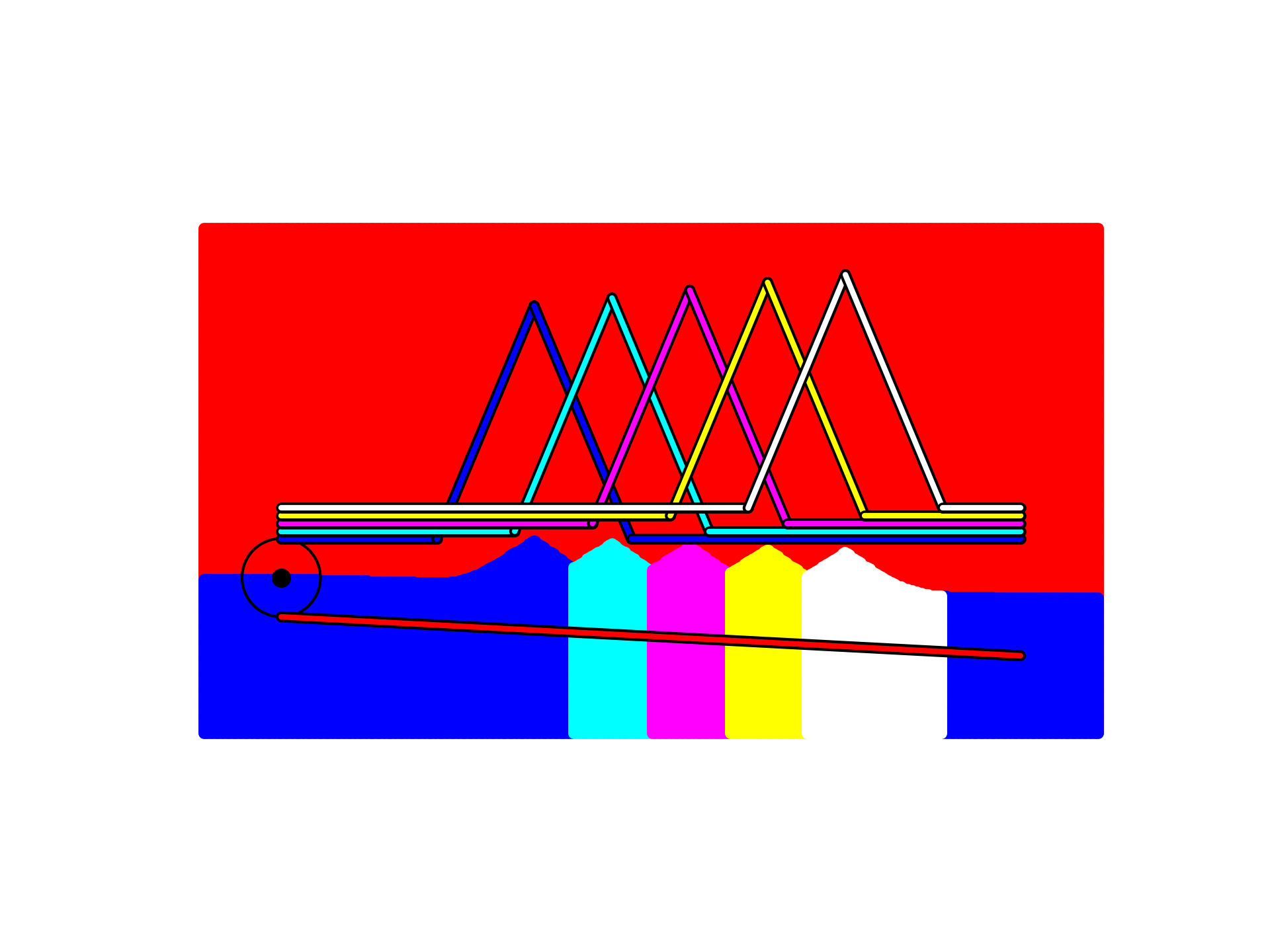}
    \caption{The construction of Lemma~\ref{lemma:traj_lp_monster}, without $T_1$. Left, a simple visualization, and right, with the Farthest Trajectory Voronoi Diagram.}
    \label{fig:minus1}
\end{figure*}

\begin{figure*}
    \centering
    \includegraphics[trim={2.4cm 2.9cm 1.85cm  2.9cm},clip,scale = 0.7]{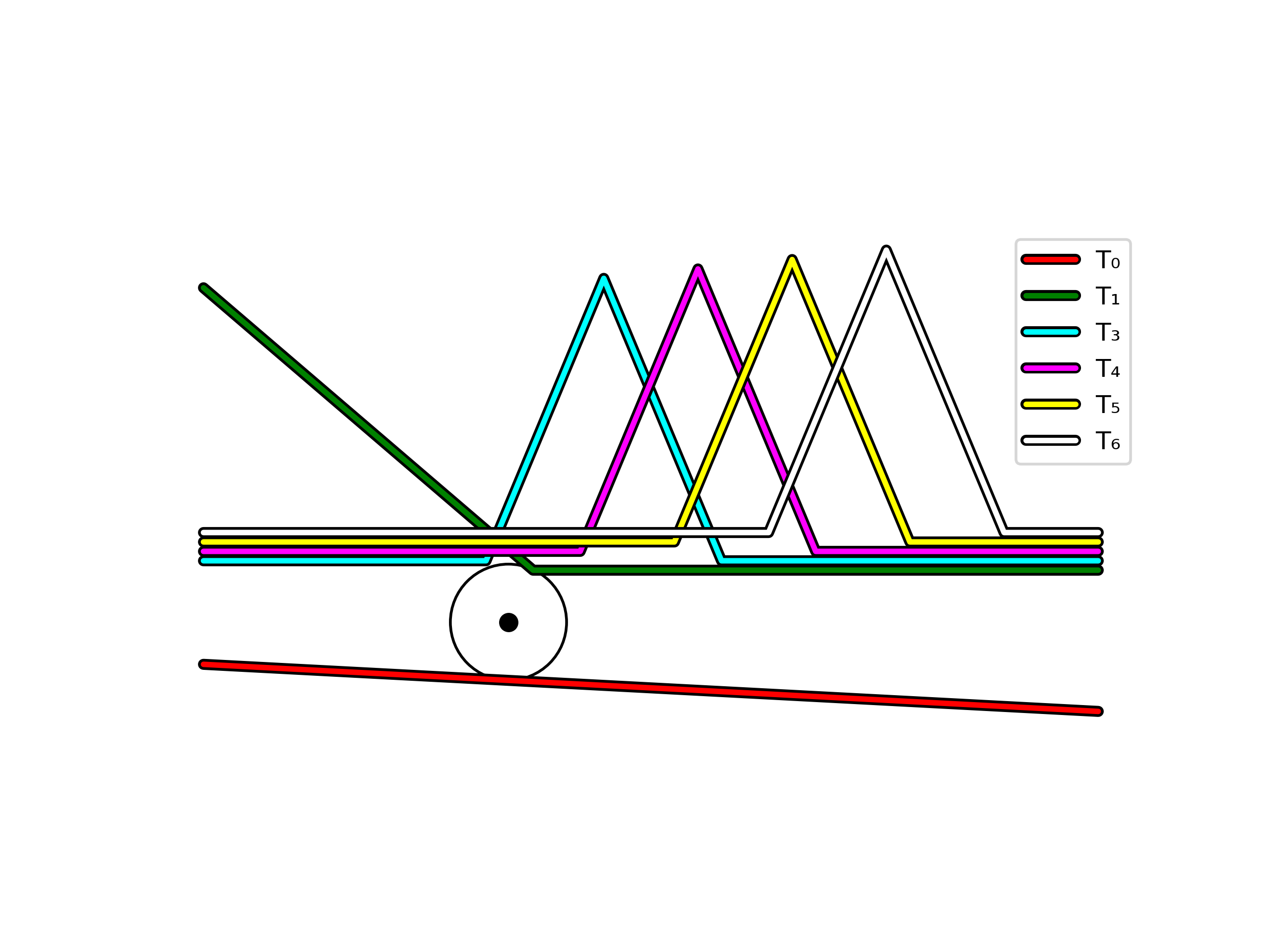}
    \includegraphics[trim={2.4cm 2.9cm 2cm 2.9cm},clip,scale = 0.7]{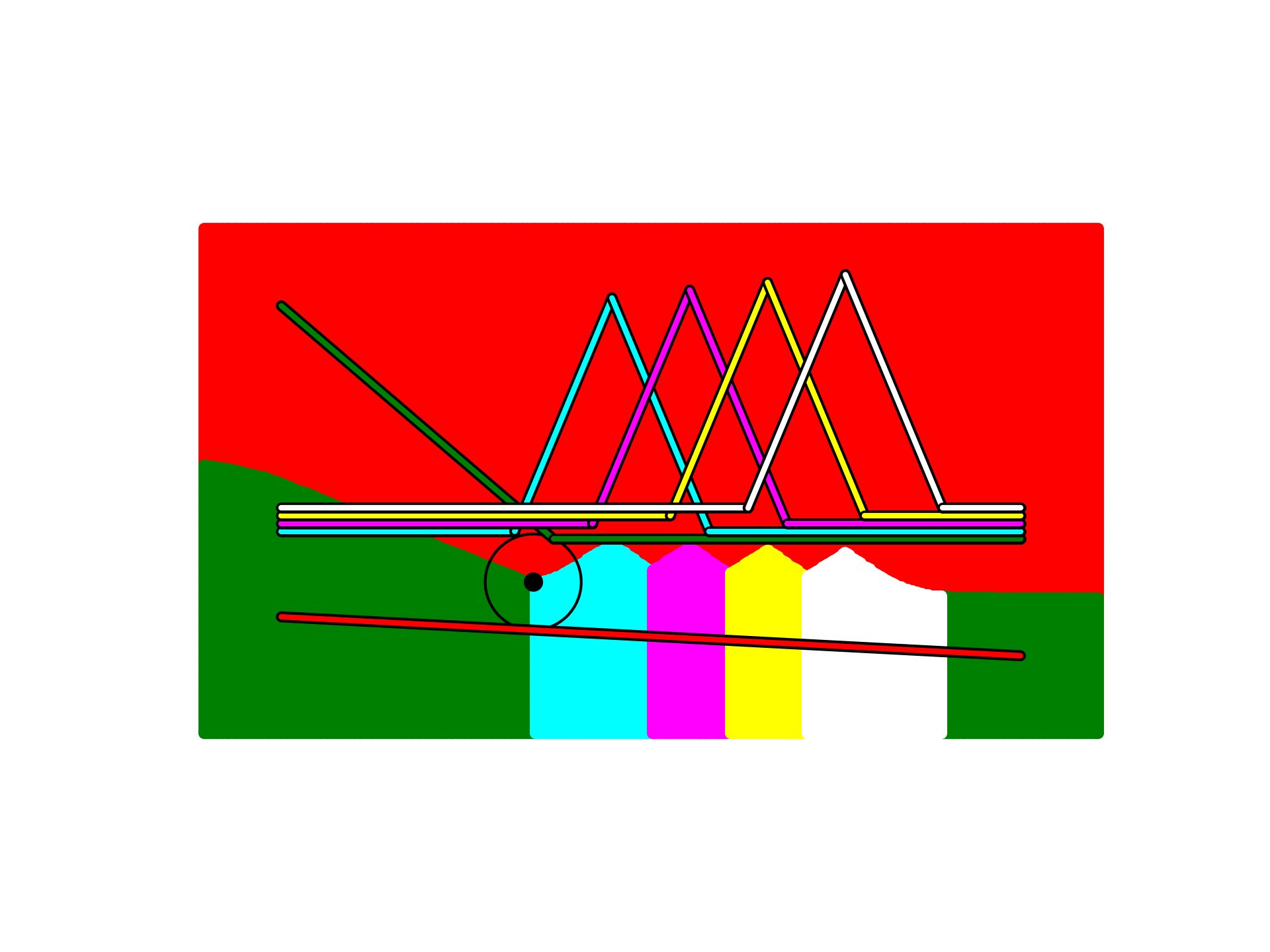}
    \caption{The construction of Lemma~\ref{lemma:traj_lp_monster}, without $T_2$. Left, a simple visualization, and right, with the FTVD.}
    \label{fig:minus2}
\end{figure*}
\begin{figure*}
    \centering
    \includegraphics[trim={2.4cm 2.9cm 1.85cm  2.9cm},clip,scale = 0.7]{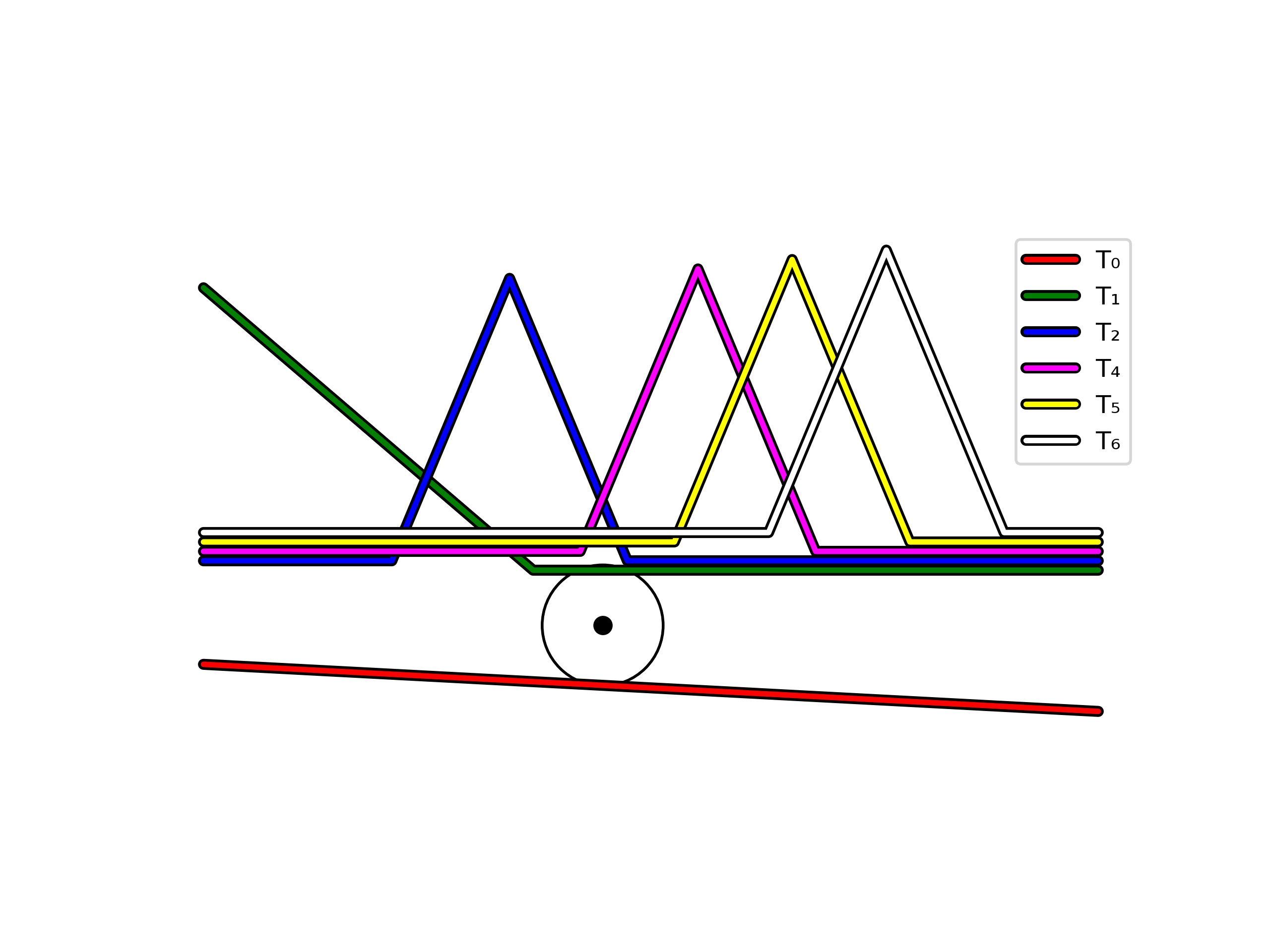}
    \includegraphics[trim={2.4cm 2.9cm 2cm 2.9cm},clip,scale = 0.7]{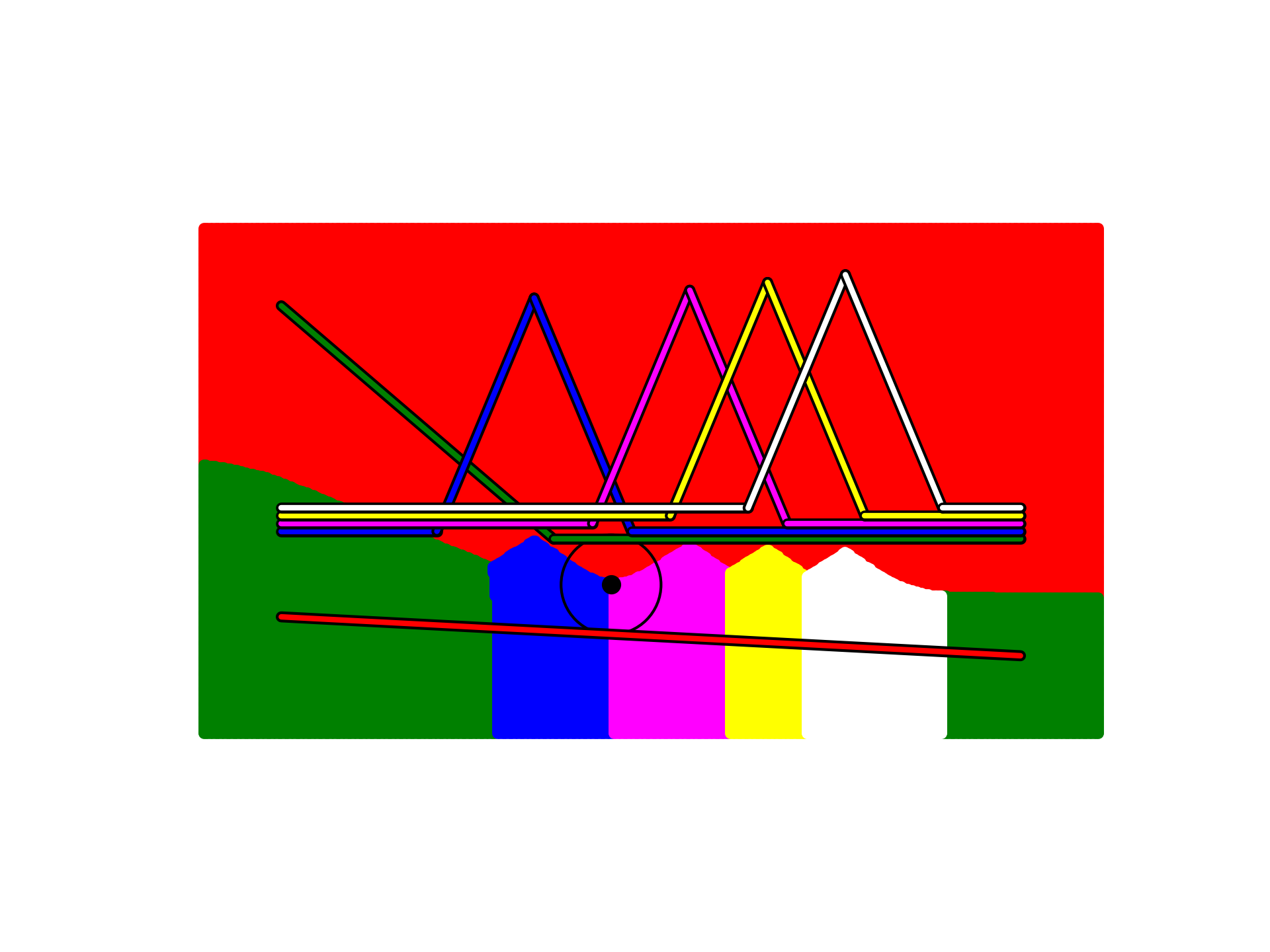}
    \caption{The construction of Lemma~\ref{lemma:traj_lp_monster}, without $T_3$. Left, a simple visualization, and right, with the FTVD.}
    \label{fig:minus3}
\end{figure*}

\begin{figure*}
    \centering
    \includegraphics[trim={2.4cm 2.9cm 1.85cm  2.9cm},clip,scale = 0.7]{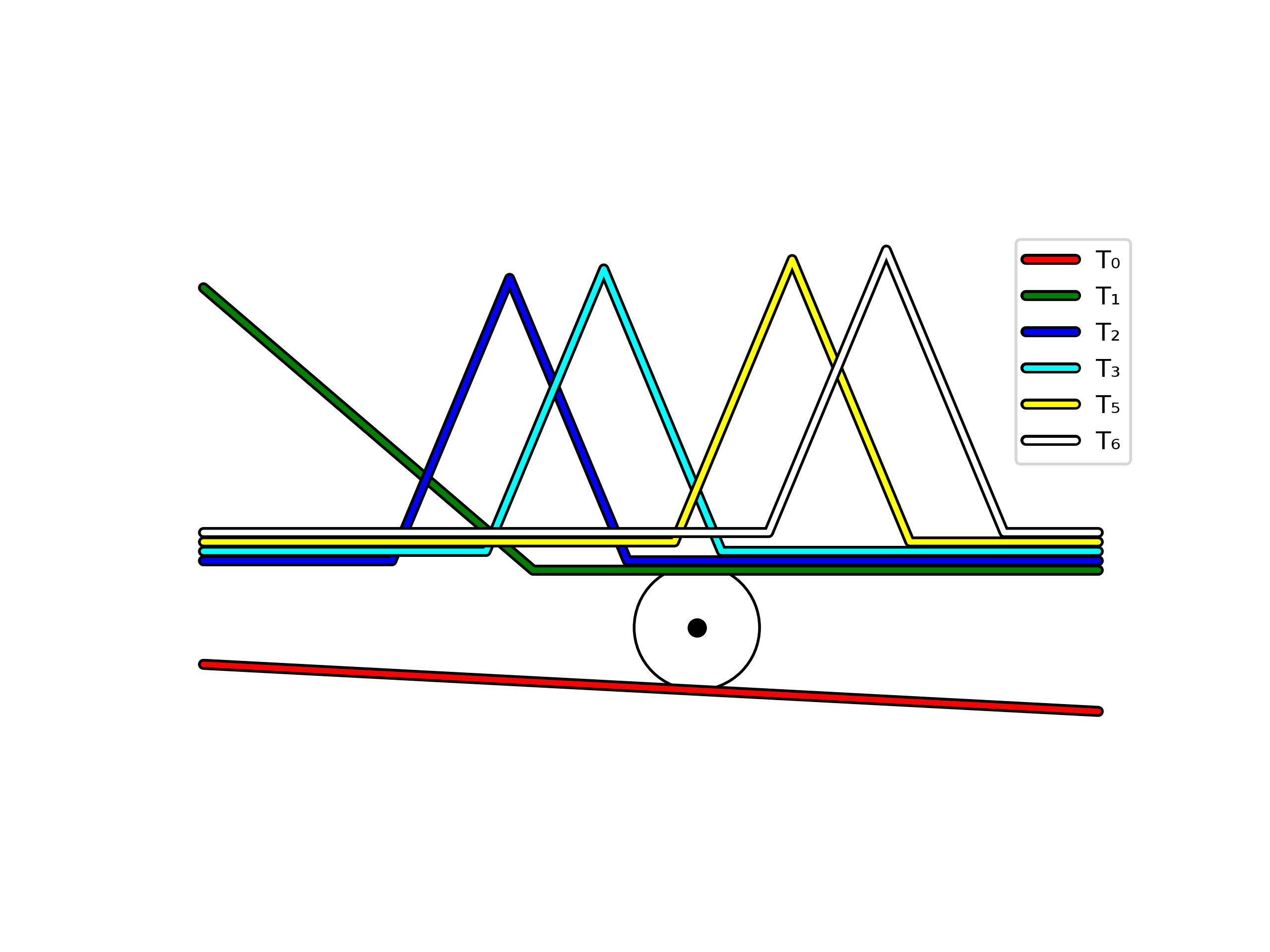}
    \includegraphics[trim={2.4cm 2.9cm 2cm 2.9cm},clip,scale = 0.7]{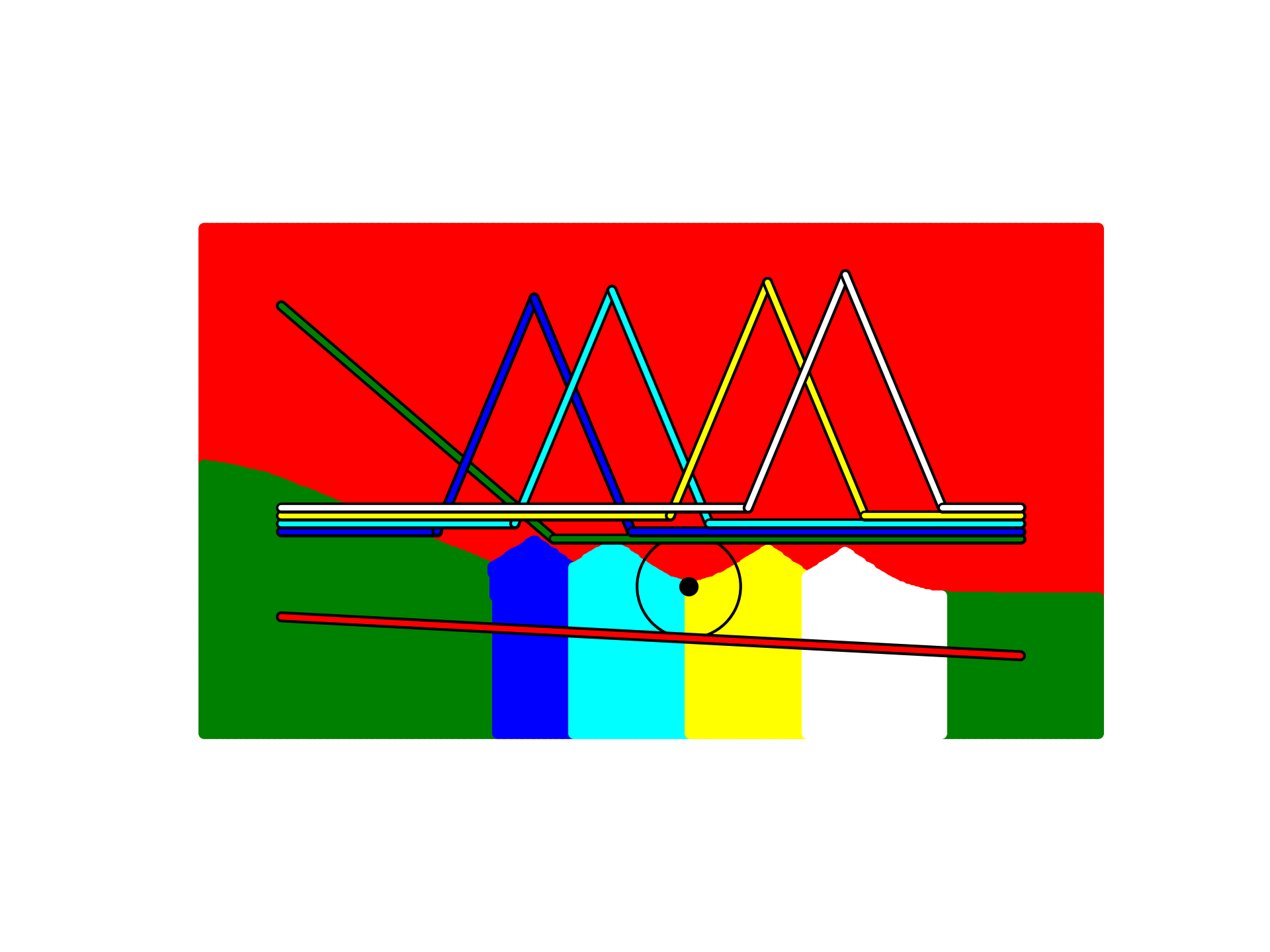}
    \caption{The construction of Lemma~\ref{lemma:traj_lp_monster}, without $T_4$. Left, a simple visualization, and right, with the FTVD.}
    \label{fig:minus4}
\end{figure*}

\begin{figure*}
    \centering
    \includegraphics[trim={2.4cm 2.9cm 1.85cm  2.9cm},clip,scale = 0.7]{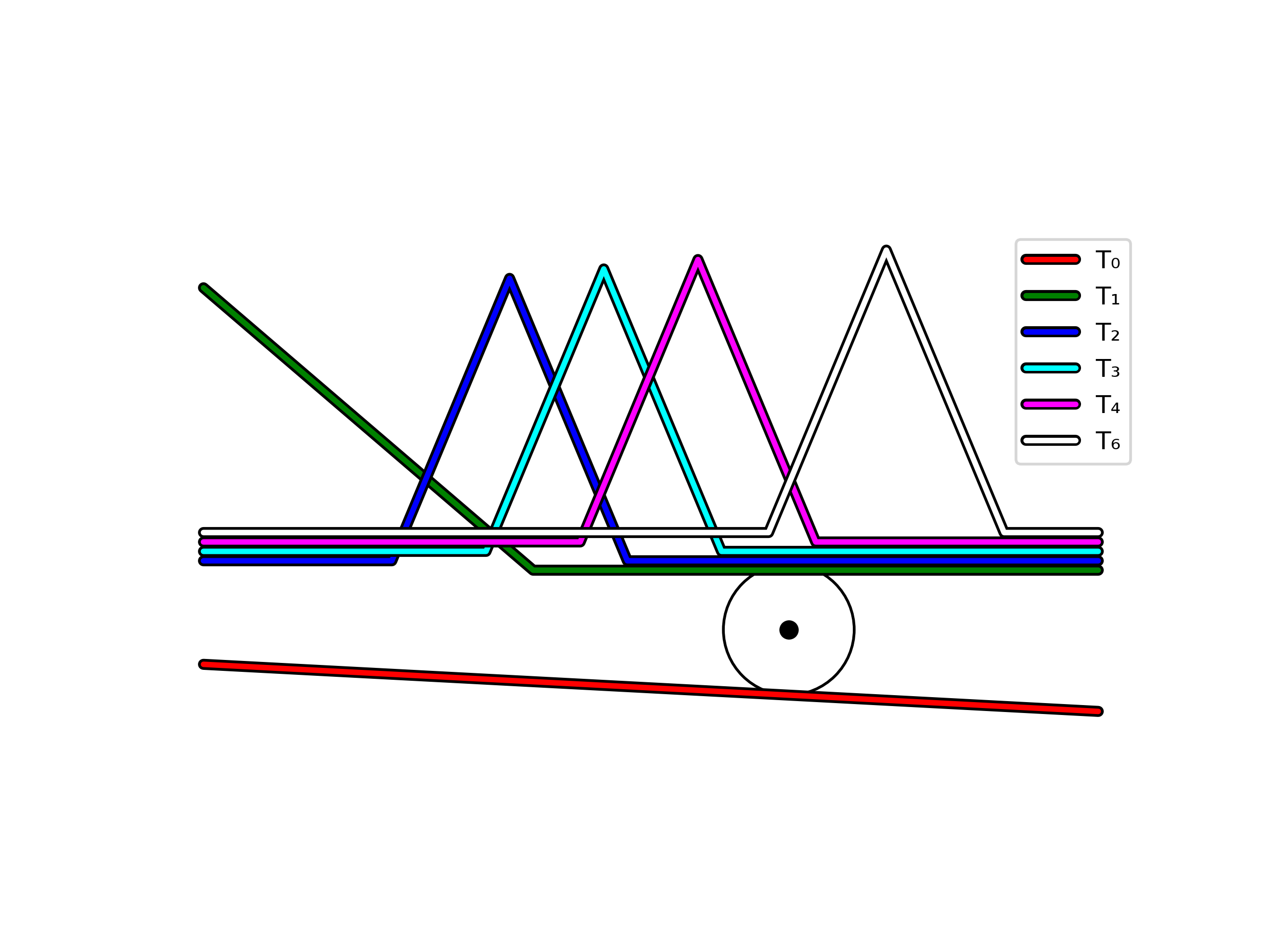}
    \includegraphics[trim={2.4cm 2.9cm 2cm 2.9cm},clip,scale = 0.7]{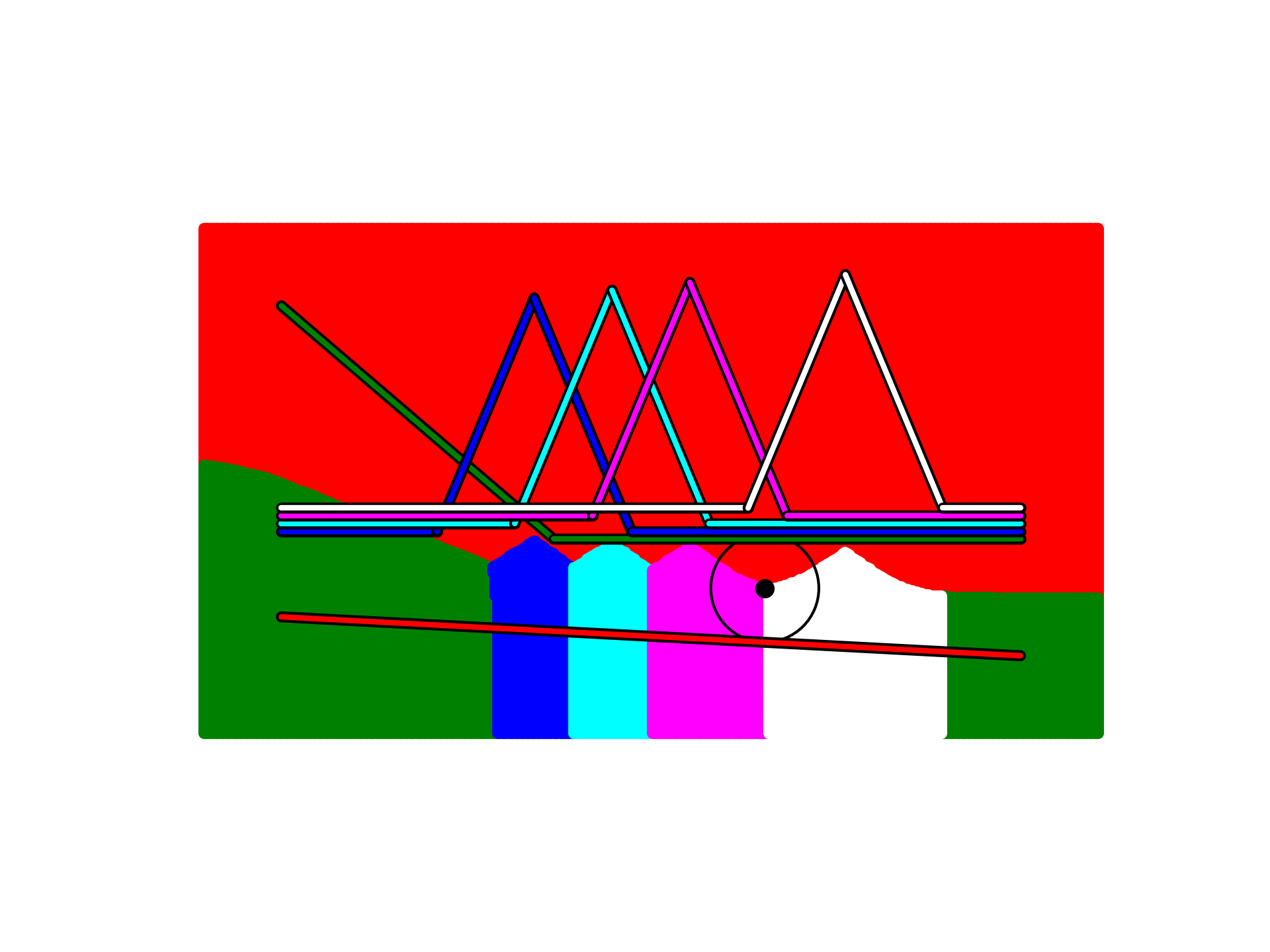}
    \caption{The construction of Lemma~\ref{lemma:traj_lp_monster}, without $T_5$. Left, a simple visualization, and right, with the FTVD.}
    \label{fig:minus5}
\end{figure*}

\begin{figure*}
    \centering
    \includegraphics[trim={2.4cm 2.9cm 1.85cm  2.9cm},clip,scale = 0.7]{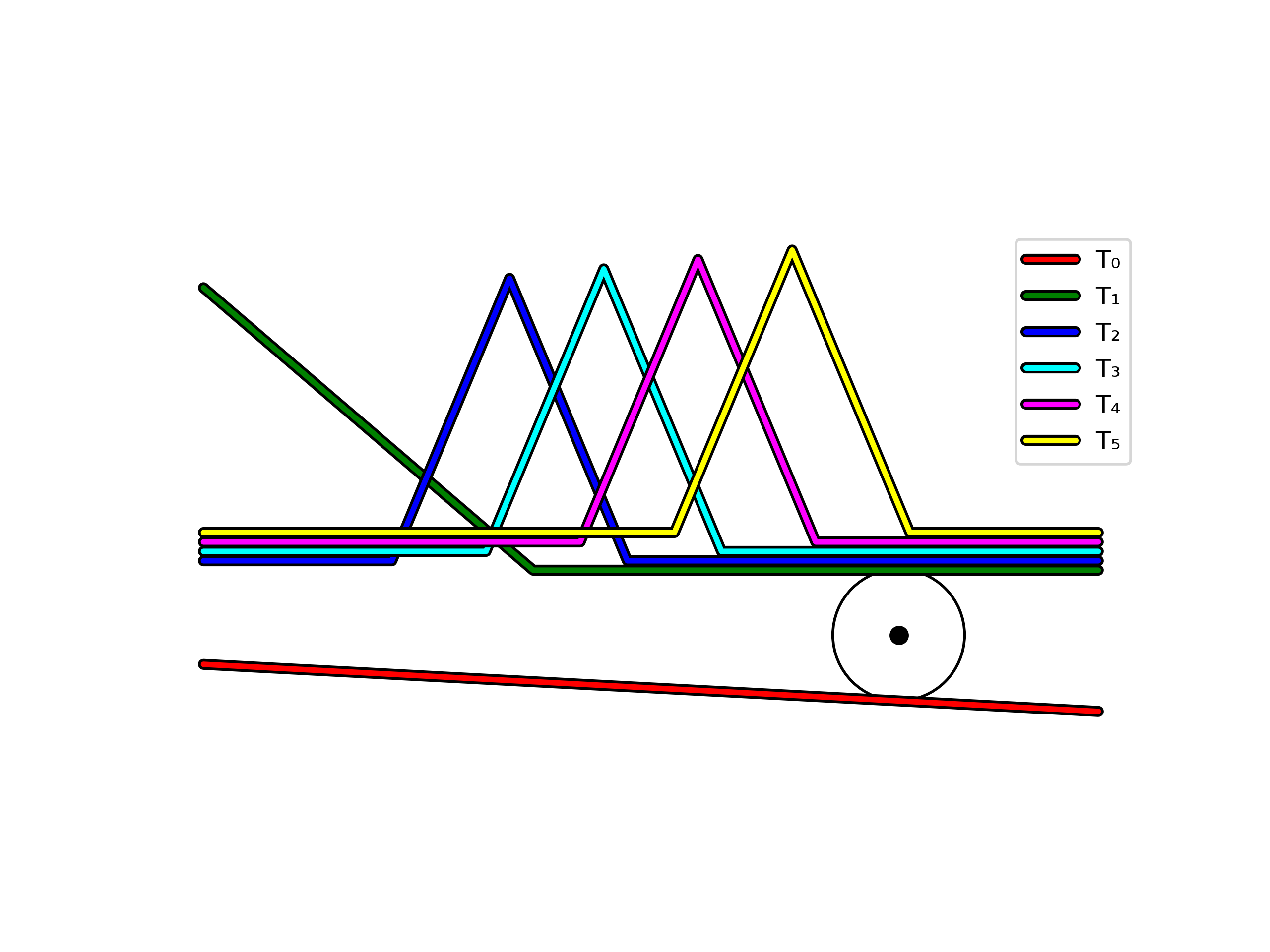}
    \includegraphics[trim={2.4cm 2.9cm 2cm 2.9cm},clip,scale = 0.7]{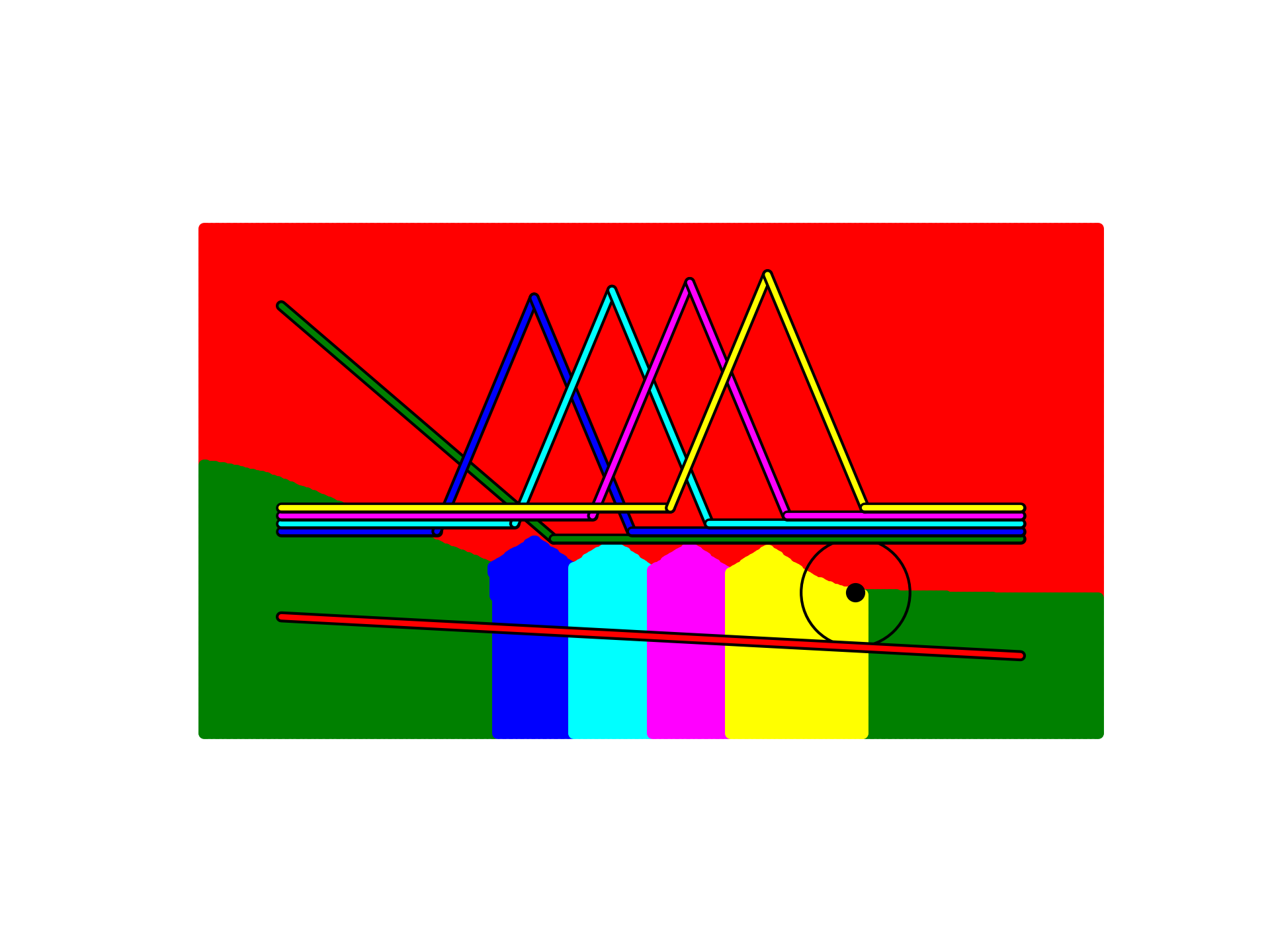}
    \caption{The construction of Lemma~\ref{lemma:traj_lp_monster}, without $T_6$. Left, a simple visualization, and right, with the FTVD.}
    \label{fig:minus6}
\end{figure*}

\end{document}